\documentclass[12pt,draftclsnofoot,journal,onecolumn]{IEEEtran}

\usepackage{ragged2e}  
\usepackage{cite}
\usepackage{balance}
\usepackage{color}
\usepackage{epsfig}
\usepackage{epstopdf}
\usepackage{graphicx}
\usepackage{tabularx}
\usepackage{booktabs}

\usepackage{varwidth}
\usepackage{enumitem} 
\usepackage[T1]{fontenc}

\usepackage{amssymb,amsmath,amsthm}
\newtheorem{lemma}{Lemma}

\usepackage{bm}
\usepackage{algorithm}%

\newfloat{routine}{htbp}{loa}
\floatname{routine}{Routine}
\makeatletter\newcommand\l@subroutine{\@dottedtocline{1}{1.5em}{2.3em}}\makeatother

\usepackage{algpseudocode}
\usepackage{pict2e}
\makeatletter
\def\BState{\State\hskip-\ALG@thistlm}

\makeatother

\newcommand{\mU}{\mathcal{U}}

\newcommand{\mS}{\mathcal{S}}

\newcommand{\mP}{\mathcal{P}}

\newcommand{\mN}{\mathcal{N}}

\newcommand{\mF}{\mathcal{F}}

\newcommand{\mX}{\mathcal{X}}
\newcommand{\mG}{\mathcal{G}}
\newcommand{\m}[1]{\mathcal{#1}}

\usepackage{colortbl}

\newcounter{remarkCounter}
\newcommand{\remark}[1]{\stepcounter{remarkCounter} {\bf{Remark \theremarkCounter}}: #1 \qed}

\newcounter{probCounter}

\begin{document}
\title{
\hspace{-0.8cm}Joint Task Offloading and Resource Allocation for \\ \hspace{-0.8cm} Multi-Server Mobile-Edge Computing Networks
\thanks{The authors are with the Department of Electrical and Computer Engineering, Rutgers University--New Brunswick, NJ, USA (e-mail: tuyen.tran@rutgers.edu, pompili@cac.rutgers.edu).}
\thanks{This work was supported by the National Science Foundation~(NSF) Grant No.~CNS-1319945.}}

\author{Tuyen~X.~Tran,~\IEEEmembership{Student Member,~IEEE} and Dario~Pompili,~\IEEEmembership{Senior Member,~IEEE}
}


\IEEEcompsoctitleabstractindextext{
\justify
\begin{abstract}
Mobile-Edge Computing~(MEC) is an emerging paradigm that provides a capillary distribution of cloud computing capabilities to the edge of the wireless access network, enabling rich services and applications in close proximity to the end users. In this article, a MEC enabled multi-cell wireless network is considered where each Base Station~(BS) is equipped with a MEC server that can assist mobile users in executing computation-intensive tasks via task offloading. The problem of Joint Task Offloading and Resource Allocation~(JTORA) is studied in order to maximize the users' task offloading gains, which is measured by the reduction in task completion time and energy consumption. The considered problem is formulated as a Mixed Integer Non-linear Program~(MINLP) that involves jointly optimizing the task offloading decision, uplink transmission power of mobile users, and computing resource allocation at the MEC servers. Due to the NP-hardness of this problem, solving for optimal solution is difficult and impractical for a large-scale network. To overcome this drawback, our approach is to decompose the original problem into (i)~a Resource Allocation~(RA) problem with fixed task offloading decision and (ii)~a Task Offloading~(TO) problem that optimizes the optimal-value function corresponding to the RA problem. We address the RA problem using convex and quasi-convex optimization techniques, and propose a novel heuristic algorithm to the TO problem that achieves a suboptimal solution in polynomial time. Numerical simulation results show that our algorithm performs closely to the optimal solution and that it significantly improves the users' offloading utility over traditional approaches. 
\end{abstract}
}

\maketitle

\IEEEdisplaynotcompsoctitleabstractindextext
\IEEEpeerreviewmaketitle

\thispagestyle{empty}


\section{Introduction}
\textbf{Motivation:}
The rapid growth of mobile applications and the Internet of Things~(IoTs) have placed severe demands on cloud infrastructure and wireless access networks such as ultra-low latency, user experience continuity, and high reliability. These stringent requirements are driving the need for highly localized services at the network edge in close proximity to the end users. 
In light of this, the Mobile-Edge Computing (MEC)~\cite{hu2015mobile} concept has emerged, which aims at uniting telco, IT, and cloud computing to deliver cloud services directly from the network edge. Differently from traditional cloud computing systems where remote public clouds are utilized, MEC servers are owned by the network operator and are implemented directly at the cellular Base Stations~(BSs) or at the local wireless Access Points~(APs) using a generic-computing platform. With this position, MEC allows for the execution of applications in close proximity to end users, substantially reducing end-to-end~(e2e) delay and releasing the burden on backhaul networks~\cite{tran2016collaborative}.


With the emergence of MEC, the ability of resource-constrained mobile devices to offload computation tasks to the MEC servers is expected to support a myriad of new services and applications such as augmented reality, IoT, autonomous vehicles and image processing. For example, the face detection and recognition application for airport security and surveillance can be highly benefit from the collaboration between mobile devices and MEC platform~\cite{soyata2012cloud}. In this scenario, a central authority such as FBI would extend their Amber alerts such that all available cell phones in the area where a missing child was last seen that opt-in to the alert would actively capture images. Due to the significant amount of processing and the need for a large database of images, the captured images are then forwarded to the MEC layer to perform face recognition.



Task offloading, however, incurs extra overheads in terms of delay and energy consumption due to the communication required between the devices and the MEC server in the uplink wireless channels. Additionally, in a system with a large number of offloading users, the finite computing resources at the MEC servers considerably affect the task execution delay~\cite{yang2015multi}. Therefore, offloading decisions and performing resource allocation become a critical problem toward enabling efficient computation offloading. Previously, this problem has been partially addressed by optimizing either the offloading decision~\cite{yang2015multi, cardellini2016game}, communication resources~\cite{chen2015decentralized, chen2016efficient}, or computing resources~\cite{yang2013framework, rahimi2013music}. Recently, Sardellitti et al.~\cite{sardellitti2015joint} addressed the joint allocation of radio and computing resources, while the authors in~\cite{lyu2016multi} considered the joint task offloading and resources optimization in a multi-user system. Both of these works, however, only concentrate on a system with a single MEC server.


\textbf{Our Vision:}
Unlike the traditional approaches mentioned above, our objective is to design a holistic solution for joint task offloading and resource allocation in a multi-server MEC-assisted network so as to maximize the users' offloading gains. Specifically, we consider a multi-cell ultra-dense network where each BS is equipped with a MEC sever to provide computation offloading services to the mobile users. The distributed deployment of the MEC servers along with the densification of (small cell) BSs---as foreseen in the 5G standardization roadmap~\cite{ge20165g}---will pave the way for real proximity, ultra-low latency access to cloud functionalities. Additionally, the benefits brought by a multi-server MEC system over the single-server MEC (aka single-cloud) system are multi-fold: 
(i)~firstly, as each MEC server may be overloaded when serving a large number of offloading users, one can release the burdens on that server by directing some users to offload to the neighboring servers from the nearby BSs, thus preventing the limited resources on each MEC server from becoming the bottle neck; (ii) secondly, each user can choose to offload its task to the BS with more favorable uplink channel condition, thus saving transmission energy consumption; (iii)~finally, coordination of resource allocation to offload users across multiple neighboring BSs can help mitigate the effect of interference and resource contention among the users and hence, improve offloading gains when multiple users offload their tasks simultaneously. 

\textbf{Challenges and Contributions:}
To exploit in full the benefits of computation offloading in the considered multi-cell, multi-server MEC network, there are several key challenges that need to be addressed. Firstly, the radio resource allocation is much more challenging than the special cases studied in the literature (cf.~\cite{lyu2016multi}) due to the presence of inter-cell interference that introduces the coupling among the achievable data rate of different users, which makes the problem nonconvex. Secondly, the complexity of the task-offloading decision is high as, for each user, one needs to decide not only whether it should offload the computation task but also which BS/server to offload the task to. Thirdly, the optimization model should take into account the inherent heterogeneity in terms of mobile devices' computing capabilities, computation task requirements, and availability of computing resources at different MEC servers. 

In this context, the main contributions of this article are summarized as follows.
\begin{itemize}
\item We model the offloading utility of each user as the weighted-sum of the improvement in task-completion time and device energy consumption; we formulate the problem of Joint Task Offloading and Resource Allocation~(JTORA) as a Mixed Integer Non-linear Program~(MINLP) that jointly optimizes the task offloading decisions, users' uplink transmit power, and computing resource allocation to offloaded users at the MEC servers, so as to maximize the system offloading utility. 
\item Given the NP-hardness of the JTORA problem, we propose to decompose the problem into (i)~a Resource Allocation~(RA) problem with fixed task offloading decision and (ii)~a Task Offloading~(TO) problem that optimizes the optimal-value function corresponding to the RA problem.
\item We further show that the RA problem can be decoupled into two independent problems, namely the Uplink Power Allocation~(UPA) problem and the Computing Resource Allocation~(CRA) problem; the resulting UPA and CRA problems are addressed using quasi-convex and convex optimization techniques, respectively.
\item We propose a novel low-complexity heuristic algorithm to tackle the TO problem and show that it achieves a suboptimal solution in polynomial time. 
\item We carry out extensive numerical simulations to evaluate the performance of the proposed solution, which is shown to be near-optimal and to improve significantly the users' offloading utility over traditional approaches. 
\end{itemize}

\textbf{Article Organization:}
The remainder of this article is organized as follows. 
In Sect.~\ref{sec:related}, we review the related works. 
In Sect.~\ref{sec:model}, we present the system model. The joint task offloading and resource allocation problem is formulated in Sect.~\ref{sec:problem}, followed by the NP-hardness proof and decomposition of the problem itself. We present our proposed solution in Sect.~\ref{sec:res_alloc} and numerical results in Sect.~\ref{sec:results}. Finally, in Sect.~\ref{sec:conclusion} we conclude the article.

\section{Related Works} \label{sec:related}
The MEC paradigm has attracted considerable attention in both academia and industry over the past several years. In 2013, Nokia Networks introduced the very first real-world MEC platform~\cite{nokiaRACS}, in which the computing platform---Radio Applications Cloud Servers~(RACS)---is fully integrated with the Flexi Multiradio BS. Saguna also introduced their fully virtualized MEC platform, so called Open-RAN~\cite{sagunaMEC}, that can provide an open environment for running third-party MEC applications. Recently, a MEC Industry Specifications Group~(ISG) was formed to standardize and moderate the adoption of MEC within the RAN~\cite{hu2015mobile}.


A number of solutions have also been proposed to exploit the potential benefits of MEC in the context of the IoTs and 5G. For instance, our previous work in~\cite{tran2016collaborative} proposed to explore the synergies among the connected entities in the MEC network and presented three representative use-cases to illustrate the benefits of MEC collaboration in 5G networks. In~\cite{tran2017wons}, we proposed a collaborative caching and processing framework in a MEC network whereby the MEC servers can perform both caching and transcoding so as to facilitate Adaptive Bit-Rate (ABR) video streaming. Similar approach was also considered in \cite{fajardo2015improving} which combined the traditional client-driven dynamic adaptation scheme, DASH, with network-assisted adaptation capabilities. In addition, MEC is also seen as a key enabling technique for connected vehicles by adding computation and geo-distributed services to the roadside BSs so as to analyze the data from proximate vehicles and roadside sensors and to propagate messages to the drivers in very low latency~\cite{nokiaCar2x}. 

Recently, several works have focused on exploiting the benefits of computation offloading in MEC network~\cite{mao2017mobile}. Note that similar problems have been investigated in conventional Mobile Cloud Computing (MCC) systems~\cite{sanaei2014heterogeneity}. However, a large body of existing works on MCC assumed an infinite amount of computing resources available in a cloudlets, where the offloaded tasks can be executed with negligible delay~\cite{zhang2013energy,zhang2015collaborative, cheng2015just}. The problem of offloading scheduling was then reduced to radio resource allocation in~\cite{chen2015decentralized} where the competition for radio resources is modeled as a congestion game of selfish mobile users. In the context of MEC, the problem of joint task offloading and resource allocation was studied in a single-user system with energy harvesting devices~\cite{mao2016dynamic}, and in a multi-cell multi-user systems~\cite{sardellitti2015joint}; however the congestion of computing resources at the MEC server was omitted. Similar problem is studied in~\cite{lyu2016multi} considering the limited edge computing resources in a single-server MEC system. 

In summary, most of the existing works did not consider a holistic approach that jointly determines the task offloading decision and the radio and computing resource allocation in a multi-cell, multi-server system as considered in this article.

\section{System Model} \label{sec:model}

\begin{figure}
\centering
\includegraphics[width=0.6\textwidth]{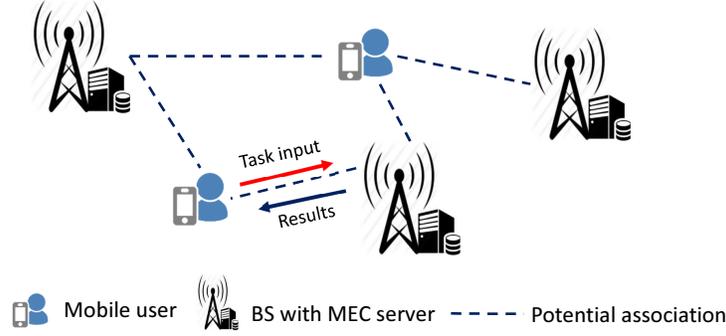}
\caption{Example of a cellular system with MEC servers deployed at the BSs.}\label{fig:MEC}
\end{figure}

We consider a multi-cell, multi-server MEC system as illustrated in Fig.~\ref{fig:MEC}, in which each BS is equipped with a MEC server to provide computation offloading services to the resource-constrained mobile users such as smart phones, tablets, and wearable devices. In general, each MEC server can be either a physical server or a virtual machine with moderate computing capabilities provisioned by the network operator and can communicate with the mobile devices through wireless channels provided by the corresponding BS. Each mobile user can choose to offload computation tasks to a MEC server from one of the nearby BSs it can connect to. We denote the set of users and MEC servers in the mobile system as $\mU = \left\{ {1,2,...,U} \right\}$ and $\mS = \left\{ {1,2,...,S} \right\}$, respectively. For ease of presentation, we will refer to the MEC server $s$, server $s$, and BS $s$ interchangeably. The modeling of user computation tasks, task uploading transmissions, MEC computation resources, and offloading utility are presented here below.

\subsection{User Computation Tasks}
We consider that each user $u\in\mU$ has one computation task at a time, denoted as $T_u$, that is atomic and cannot be divided into subtasks. 
Each computation task $T_u$ is characterized by a tuple of two parameters, $\left\langle {{d_u},{c_u}} \right\rangle$, in which $d_u~\rm[bits]$ specifies the amount of input data necessary to transfer the program execution (including system settings, program codes, and input parameters) from the local device to the MEC server, and $c_u~\rm[cycles]$ specifies the workload, i.e., the amount of computation to accomplish the task. The values of $d_u$ and $c_u$ can be obtained through carefully profiling of the task execution~\cite{miettinen2010energy, yang2015multi}. Each task can be performed locally on the user device or offloaded to a MEC server. By offloading the computation task to the MEC server, the mobile user would save its energy for task execution; however, it would consume additional time and energy for sending the task input in the uplink.

Let $f_u^l > 0$ denote the local computing capability of user $u$ in terms of CPU $\rm{cycles/s}$. Hence, if user $u$ executes its task locally, the task completion time is $t_u^l = \frac{c_u}{f_u^l}~\rm[seconds]$. 
To calculate the energy consumption of a user device when executing its task locally, we use the widely adopted model of the energy consumption per computing cycle as $\m{E} = \kappa {f^2}$~\cite{chen2015decentralized, wen2012energy}, where $\kappa$ is the energy coefficient depending on the chip architecture and $f$ is the CPU frequency. Thus, the energy consumption, $E_{u}^l\left[ J \right]$, of user $u$ when executing its task $T_u$ locally, is calculated as,
\begin{equation}
E_u^l = \kappa {\left( {f_{u}^l} \right)^2} c_u.
\end{equation}

\subsection{Task Uploading}
In case user $u$ offloads its task $T_u$ to one of the MEC servers, the incurred delay comprises: (i)~the time $t_{\text{up}}^u~\rm[s]$ to transmit the input to the MEC server on the uplink, (ii)~the time $t_{\text{exe}}^u~\rm[s]$ to execute the task at the MEC server, and (iii)~the time to transmit the output from the MEC server back to the user on the downlink. Since the size of the output is generally much smaller than the input, plus the downlink data rate is much higher than that of the uplink, we omit the delay of transferring the output in our computation, as also considered in~\cite{lyu2016multi, chen2015decentralized}. 


In this work, we consider the system with OFDMA as the multiple access scheme in the uplink~\cite{dahlman20134g}, in which the operational frequency band $B$ is divided into $N$ equal sub-bands of size $W = B/N~\rm[Hz]$. To ensure the orthogonality of uplink transmissions among users associated with the same BS, each user is assigned to one sub-band. Thus, each BS can serve at most $N$ users at the same time. Let $ \mN = \left\{ {1,...,N} \right\}$ be the set of available sub-band at each BS. We define the task offloading variables, which also incorporate the uplink sub-band scheduling, as $x_{us}^j, u \in \mU, s\in \mS, j\in \mN$, where $x_{us}^j = 1$ indicates that task $T_u$ from user $u$ is offloaded to BS $s$ on sub-band $j$, and $x_{us}^j= 0$ otherwise. We define the ground set $\mG$ that contains all the task offloading variables as $\mG = \left\{ {x_{us}^j\left| u \in \mU, s\in \mS, j\in \mN \right.} \right\}$ and the task offloading policy $\mX$ expressed as $\mX = \left\{ {x_{us}^j \in \mG \left| {x_{us}^j = 1} \right.} \right\}$.
As each task can be either executed locally or offloaded to at most one MEC server, a feasible offloading policy must satisfy the constraint below, 
\begin{equation} \label{eq:off_con}
\sum\limits_{s \in \mS} {\sum\limits_{j \in \mN} {x_{us}^j} }  \le 1,, \forall u \in \mU.
\end{equation}
Additionally, we denote ${\mU_s} = \left\{ {u \in \mU\left| {\sum\nolimits_{j \in \mN} {x_{us}^j}  = 1} \right.} \right\}$ as the set of users offloading their tasks to server $s$, and ${\mU_{{\rm{off}}}} = \bigcup\limits_{s \in \mS} {{\mU_s}}$ as the set of users that offload their tasks.

Furthermore, we consider that each user and BS have a single antenna for uplink transmissions (as also considered in~\cite{saad2014college,lyu2016multi}). Extension to the case where each BS uses multiple antennas for receiving uplink signals will be addressed in a future work. Denote $h_{us}^j$ as the uplink channel gain between user $u$ and BS $s$ on sub-band $j$, which captures the effect of path-loss, shadowing, and antenna gain. Note that the user-BS association usually takes place in a large time scale (duration of an offloading session) that is much larger than the time scale of small-scale fading. Hence, similar to~\cite{ye2013user}, we consider that the effect of fast-fading is averaged out during the association. Let $\mP = \left\{ {{p_u}\left| {0 < {p_u} \le {P_u},u \in \mU_{\rm{off}} } \right.} \right\}$ denote the users' transmission power, where $p_u~\rm[W]$ is the transmission power of user $u$ when uploading its task's input $I_u$ to the BS, subject to a maximum budget $P_u$. Note that $p_u = 0, \forall u \notin \mU_{\rm{off}}$.
As the users transmitting to the same BS use different sub-bands, the uplink intra-cell interference is well mitigated; still, these users suffer from the inter-cell interference. In this case, the Signal-to-Interference-plus-Noise Ratio~(SINR) from user $u$ to BS $s$ on sub-band $j$ is given by,  
\begin{equation} \label{eq:SINR}
{\gamma _{us}^j} = \frac{{{p_u}{h_{us}^j}}}{{\sum\limits_{k \in \mU\backslash \mU_s} {x_{ks}^j{p_k}{h_{ks}^j}}  + {\sigma ^2}}}, \forall u \in \mU, s \in \mS, j \in \mN,
\end{equation}
where $\sigma^2$ is the background noise variance and the first term at the denominator is the accumulated intra-cell interference from all the users associated with other BSs on the same sub-band $j$. Since each user only transmits on one sub-band, the achievable rate $\rm[bits/s]$ of user $u$ when sending data to BS $s$ is given as,
\begin{equation} \label{eq:R_us}
{R_{us}} = W{\log _2}\left( {1 + \gamma_{us}} \right),
\end{equation}
where ${\gamma _{us}} = \sum\nolimits_{j \in \mN} {\gamma _{us}^j} $. Moreover, let ${x _{us}} = \sum\nolimits_{j \in \mN} {x_{us}^j}, \forall u \in \mU, s \in \mS $. Hence, the transmission time of user $u$ when sending its task input $d_u$ in the uplink can be calculated as,
\begin{equation} \label{eq:t_up}
t_{{\rm{up}}}^u = \sum\limits_{s \in \mS} {\frac{{{x_{us}}{d_u}}}{{{R_{us}}}}} , \forall u \in \mU.
\end{equation}

\subsection{MEC Computing Resources}
The MEC server at each BS is able to provide computation offloading service to multiple users concurrently. The computing resources made available by each MEC server to be shared among the associating users are quantified by the computational rate $f_s$, expressed in terms of number of CPU $\rm{cycles/s}$. 
After receiving the offloaded task from a user, the server will execute the task on behalf of the user and, upon completion, will return the output result back to the user. We define the computing resource allocation policy as $\m{F} = \left\{ {{f_{us}}\left| u \in \mU, s\in \mS\right.} \right\}$, in which $f_{us}~\rm[cycles/s] > 0$ is the amount of computing resource that BS $s$ allocates to task $T_u$ offloaded from user $u \in \mU_s$. Hence, clearly $f_{us} = 0, \forall u \notin \mU_s$. In addition, a feasible computing resource allocation policy must satisfy the computing resource constraint, expressed as,
\begin{equation}
\sum\limits_{u \in \mU} {{f_{us}}}  \leq  {f_s}, \forall s \in \mS.
\end{equation}
Given the computing resource assignment $\left\{ {{f_{us}},s \in \mS} \right\}$, the execution time of task $T_u$ at the MEC servers is,
\begin{equation} \label{eq:d_exe}
t_{{\rm{exe}}}^u = \sum\limits_{s \in \mS} {\frac{{{x_{us}}{c_u}}}{{{f_{us}}}}}, \forall u \in \mU.
\end{equation}

\subsection{User Offloading Utility}
Given the offloading policy $\mX$, the transmission power $p_u$, and the computing resource allocation $f_{us}$'s, the total delay experienced by user $u$ when offloading its task is given by,
\begin{equation} \label{eq:tu}
{t_u} = t_{{\rm{up}}}^u + t_{{\rm{exe}}}^u = \sum\limits_{s \in \mS} {{x_{us}}\left( {\frac{{{d_u}}}{{{R_{us}}}} + \frac{{{c_u}}}{{{f_{us}}}}} \right)}, \forall u \in \mU.
\end{equation}
The energy consumption of user $u$, $E_u\left[J \right]$, due to uploading transmission is calculated as $E_u = \frac{{{p_u}t_{{\rm{up}}}^u}}{{{\xi _u}}}, \forall u \in \mU$, where $\xi_u$ is the power amplifier efficiency of user $u$. Without loss of generality, we assume that $\xi_u = 1, \forall u \in \mU$. Thus, the uplink energy consumption of user $u$ simplifies to,
\begin{equation} \label{eq:energy}
E_u = p_ut_{\rm{up}}^u = {p_u}{d_u}\sum\limits_{s \in \mS} {\frac{{{x_{us}}}}{{{R_{us}}}}}  , \forall u \in \mU.
\end{equation}

In a mobile cloud computing system, the users' QoE is mainly characterized by their task completion time and energy consumption. In the considered scenario, the relative improvement in task completion time and energy consumption are characterized by ${\frac{{t_u^l - {t_u}}}{{t_u^l}}}$ and ${\frac{{E_u^l - {E_u}}}{{E_u^l}}}$, respectively~\cite{lyu2016multi}. Therefore, we define the offloading utility of user $u$ as,
\begin{equation} \label{eq:user_util}
{J_u} = \left( {\beta _u^t\frac{{t_u^l - {t_u}}}{{t_u^l}} + \beta _u^e\frac{{E_u^l - {E_u}}}{{E_u^l}}} \right) \sum\limits_{s \in \mS} {{x_{us}}} , \forall u \in \mU,
\end{equation}
in which $\beta _u^t,\beta _u^e \in \left[ {0,1} \right]$, with $\beta_u^t + \beta_u^e = 1, \forall u \in \mU$, specify user $u$'s preference on task completion time and energy consumption, respectively. For example, a user $u$ with short battery life can increase  $\beta _u^e$ and decrease $\beta _u^t$ so as to save more energy at the expense of longer task completion time. Note that offloading too many tasks to the MEC servers will cause excessive delay due to the limited bandwidth and computing resources at the MEC servers, and consequently degrade some users' QoE compared to executing their tasks locally. Hence, clearly user $u$ should not offload its task to the MEC servers if $J_u \leq 0$. 

The expressions of the task completion time and energy consumption in~\eqref{eq:user_util} clearly shows the interplay between radio access and computational aspects, which motivates a joint optimization of offloading scheduling, radio, and computing resources so as to optimize users' offloading utility. 

\section{Problem Formulation} \label{sec:problem}
We formulate here the problem of joint task offloading and resource allocation, followed by the outline of our decomposition approach.

\subsection{Joint Task Offloading and Resource Allocation Problem}
For a given offloading decision $\mX$, uplink power allocation $\mP$, and computing resource allocation $\mF$, we define the system utility as the weighted-sum of all the users' offloading utilities,
\begin{equation} \label{eq:sys_util}
J\left( {\m{X}, \mP,\mF} \right) = \sum\nolimits_{u \in \mU} {{\lambda _u}{J_u}},
\end{equation}
with $J_u$ given in~\eqref{eq:user_util} and ${\lambda _u} \in \left( {0,1} \right]$ specifying the resource provider's preference towards user $u$, $\forall u \in \mU$. For instance, depending on the payments offered by the users, the resource provider could prioritize users with higher revenues for offloading by increasing their corresponding preferences. With this position, we formulate the Joint Task Offloading and Resource Allocation~(JTORA) problem as a system utility maximization problem, i.e.,
\begin{subequations} \label{eq:prob}
\begin{align} \label{eq:prob_a}
\mathop {\max }\limits_{\m{X},\mP,\mF} \hspace{0.2cm} &J\left( {\m{X}, \mP,\mF} \right) \\ \label{eq:prob_b}
\rm{s.t.} \hspace{0.2cm} &{x_{us}^j} \in \left\{ {0,1} \right\}, \forall u \in \mU, s \in \mS, j \in \mN,\\ \label{eq:prob_c}
&\sum\limits_{s \in \mS} {\sum\limits_{j \in \mN} {x_{us}^j} }  \le 1, \forall u \in \mU, \\ \label{eq:prob_d}
&\sum\limits_{u \in \mU} {{x_{us}^j}}  \le 1, \forall s \in \mS, j \in \mN,\\ \label{eq:prob_e}
&0 < p_u \leq P_u, \forall u \in \mU_{\rm{off}}, \\ \label{eq:prob_f}
&f_{us} > 0, \forall u \in \mU_s, s\in \mS, \\ \label{eq:prob_g}
&\sum\nolimits_{u \in \mU} {{f_{us}}}  \le {f_s}, \forall s \in \mS.
\end{align}
\end{subequations}
The constraints in the formulation above can be explained as follows: constraints~\eqref{eq:prob_b} and \eqref{eq:prob_c} imply that each task can be either executed locally or offloaded to at most one server on one sub-band; constraint~\eqref{eq:prob_d} implies that each BS can serve at most one user per sub-band; constraint~\eqref{eq:prob_e} specifies the transmission power budget of each user; finally, constraints~\eqref{eq:prob_f} and \eqref{eq:prob_g} state that each MEC server must allocate a positive computing resource to each user associated with it and that the total computing resources allocated to all the associated users must not excess the server's computing capacity. 
The JTORA problem in~\eqref{eq:prob} is a Mixed Integer Nonlinear Program~(MINLP), which can be shown to be NP-hard; hence, finding the optimal solution usually requires exponential time complexity~\cite{pochet2006production}. Given the large number of variables that scale linearly with the number of users, MEC servers, and sub-bands, our goal is to design a low-complexity, suboptimal solution that achieves competitive performance while being practical to implement.

\subsection{Problem Decomposition}
Given the high complexity of the JTORA problem due to the combinatorial nature of the task offloading decision, our approach in this article is to temporarily fix the task offloading decision $\mX$ and to address the resulting problem, referred to as the Resource Allocation~(RA) problem. The solution of the RA problem will then be used to derive the solution of the original JTORA problem. The decomposition process is described as follows.
Firstly, we rewrite the JTORA problem in (\ref{eq:prob}) as,
\begin{subequations} \label{eq:prob1}
\begin{align} \label{eq:prob1_a}
&\mathop {\max }\limits_{\mX} \left( {\mathop {\max }\limits_{\mP,\mF} J\left( {\m{X}, \mP,\mF} \right)  } \right) \\ \label{eq:prob1_b}
&\rm{s.t.} \hspace{0.2cm} \eqref{eq:prob_b}-\eqref{eq:prob_g}. 
\end{align}
\end{subequations}
Note that the constraints on the offloading decision, $\m{X}$, in~\eqref{eq:prob_b}, \eqref{eq:prob_c}, \eqref{eq:prob_d}, and the RA policies, $\mP, \mF$, in~\eqref{eq:prob_e}, \eqref{eq:prob_f}, \eqref{eq:prob_g}, are decoupled from each other; therefore, solving the problem in~\eqref{eq:prob1} is equivalent to solving the following Task Offloading~(TO) problem,
\begin{subequations} \label{eq:master}
\begin{align} \label{eq:master_a}
&\mathop {\max }\limits_{\mX} J^*\left( {\m{X}} \right)  \\ \label{eq:master_b}
&\rm{s.t.} \hspace{0.2cm} \eqref{eq:prob_b}, \eqref{eq:prob_c}, \eqref{eq:prob_d},
\end{align}
\end{subequations}
in which $J^*\left( {\m{X}} \right)$ is the optimal-value function corresponding to the RA problem, written as,
\begin{subequations} \label{eq:prob_RA}
\begin{align} \label{eq:prob_RA_a}
& J^*\left( {\m{X}} \right) = \mathop {\max }\limits_{\mP,\mF}  J\left( {\mX,\mP,\mF} \right)  \\ \label{eq:prob_RA_b}
&\rm{s.t.} \hspace{0.2cm} \eqref{eq:prob_e}, \eqref{eq:prob_f}, \eqref{eq:prob_g},
\end{align}
\end{subequations}
In the next section, we will present our solutions to both the RA problem and the TO problem so as to finally obtain the solution to the original JTORA problem.

\section{Low-complexity Algorithm for Joint Task Offloading and Resource Allocation} \label{sec:res_alloc}
We present now our low-complexity approach to solve the JTORA problem by solving first the RA problem in~\eqref{eq:prob_RA} and then using its solution to derive the solution of the TO problem in~\eqref{eq:master}. 


Firstly, given a feasible task offloading decision $\mX$ that satisfies constraints~\eqref{eq:prob_b}, \eqref{eq:prob_c}, and \eqref{eq:prob_d}, and using the expression of $J_u$ in~\eqref{eq:user_util}, the objective function in~\eqref{eq:prob_RA_a} can be rewritten as,
\begin{equation} \label{eq:sys_util_new}
J\left( {\mX,\mP,\mF} \right) = \sum\limits_{s \in \mS} {\sum\limits_{u \in {\mU_s}} {{\lambda _u}\left( {\beta _u^t + \beta _u^e} \right)} }  - V\left( {\mX,\mP,\mF} \right),
\end{equation}
\begin{equation} \label{eq:V}
\hspace{-.00cm} \text{where} \hspace{.4cm} \left( {\mX,\mP,\mF} \right) = \sum\limits_{s \in \mS} {\sum\limits_{u \in {\mU_s}} {{\lambda _u}\left( {\frac{{\beta _u^t{t_u}}}{{t_u^l}} + \frac{{\beta _u^e{E_u}}}{{E_u^l}}} \right)} }.
\end{equation}
%
%
We observe that the first term on the right hand side~(RHS) of \eqref{eq:sys_util_new} is constant for a particular offloading decision, while $V\left( {\mX,\mP,\mF} \right)$ can be seen as the total offloading overheads of all offloaded users. Hence, we can recast~\eqref{eq:prob_RA} as the problem of minimizing the total offloading overheads, i.e.,
\begin{subequations} \label{eq:p_overhead}
\begin{align} 
&\mathop {\min }\limits_{\mP,\mF} V\left( {\mX,\mP,\mF} \right)   \\ \label{eq:p_overhead_b}
&\rm{s.t.} \hspace{0.2cm} \eqref{eq:prob_e}, \eqref{eq:prob_f}, \eqref{eq:prob_g}.
\end{align}
\end{subequations}
Furthermore, from~\eqref{eq:tu}, \eqref{eq:energy}, and \eqref{eq:V}, we have,
\begin{equation} \label{eq:V_new}
\hspace{-0.1cm} V\left( {\mX,\mP,\mF} \right) = \sum\limits_{s \in \mS} {\sum\limits_{u \in {\mU_s}} {\frac{{{\phi _u} + {\psi _u}{p_u}}}{{{{\log }_2}\left( {1 + {\gamma _{us}}} \right)}}} }  + \sum\limits_{s \in \mS} {\sum\limits_{u \in {\mU_s}} {\frac{{{\eta _u}}}{{{f_{us}}}}} },
\end{equation}
in which, for simplicity, ${\phi _u} = \frac{{{\lambda _u}\beta _u^t{d_u}}}{{t_u^lW}}$, ${\psi _u} = \frac{{{\lambda _u}\beta _u^e{d_u}}}{{E_u^lW}}$, and ${\eta _u} = {\lambda _u}\beta _u^tf_u^l$. Notice from~\eqref{eq:p_overhead_b} and \eqref{eq:V_new} that the problem in \eqref{eq:p_overhead} has a \emph{separable structure}, i.e., the objectives and constraints corresponding to the power allocation $p_u$'s and computing resource allocation $f_{us}$'s can be decoupled from each other. Leveraging this property, we can decouple problem~\eqref{eq:p_overhead} into two independent problems, namely the \emph{Uplink Power Allocation~(UPA)} and the \emph{Computing Resource Allocation~(CRA)}, and address them separately, as described in the following sections.

\subsection{Uplink Power Allocation~(UPA)}
The UPA problem is decoupled from problem~\eqref{eq:p_overhead} by considering the first term on the RHS of~\eqref{eq:V_new} as the objective function. Specifically, the UPA problem is expressed as,
\begin{subequations} \label{eq:UPA}
\begin{align}  \label{eq:UPA_a}
&\mathop {\min }\limits_{\mP} \sum\limits_{s \in \mS} {\sum\limits_{u \in {\mU_s}} {\frac{{{\phi _u} + {\psi _u}{p_u}}}{{{{\log }_2}\left( {1 + {\gamma _{us}}} \right)}}} }   \\ \label{eq:UPA_b}
&{\rm{s.t.}} \hspace{0.2cm} 0 < p_u \leq P_u, \forall u \in \mU.
\end{align}
\end{subequations}
Problem~\eqref{eq:UPA} is non-convex and difficult to solve because the uplink SINR $\gamma_{us}^j$ corresponding to user $u\in \mU_s$ depends on the transmit power of the other users associated with other BSs on the same sub-band $j$ through the inter-cell interference
$I_{us}^j = \sum\limits_{w \in \mS\backslash \left\{ s \right\}} {\sum\limits_{k \in {\mU_w}} {x_{ks}^j{p_k}{h_{ks}^j}} }$, 
as seen in~\eqref{eq:SINR}. Our approach is to find an approximation for $I_{us}^j$ and thus for ${\gamma_{us}^j}$ such that problem~\eqref{eq:UPA} can be decomposed into sub-problems that, in turn, can be efficiently solved. The optimal uplink power allocation $\mP^*$ still generates small objective value for~\eqref{eq:UPA}. Suppose each BS $s\in \mS$ calculates its uplink power allocation independently, i.e., without mutual cooperation, and informs its associated users about the uplink transmit power; then, an achievable upper bound for $I_{us}^j$ is given by,
\begin{equation} \label{eq:Iapprox}
\tilde I_{us}^j \buildrel \Delta \over = \sum\limits_{w \in \mS \backslash \left\{ s \right\}} {\sum\limits_{k \in {\mU_w}} {x_{ks}^j {P_k}{h_{ks}^j}} } , \forall u \in \mU_s, s\in \mS, j \in \mN.
\end{equation}
Similar to~\cite{du2014wireless}, we argue that $\tilde I_{us}^j$ is a good estimate of $I_{us}^j$ since our offloading decision $\mX$ is geared towards choosing the appropriate user-BS associations so as that $\tilde I_{us}^j$ be small in the first place. This means that a small error in $I_{us}^j$ should not lead to large bias in $\gamma_{us}^j$~\cite{du2014wireless}. 

By replacing $I_{us}^j$ with $\tilde I_{us}^j$, we get the approximation for the uplink SINR for user $u$ uploading to BS $s$ on sub-band $j$ as,
\begin{equation}
{{\tilde \gamma }_{us}^j} = \frac{{{p_u}{h_{us}^j}}}{{\tilde I_{us}^j + {\sigma ^2}}}, \forall u \in \mU_s, s \in \mS, j\in \mN.
\end{equation}
Let ${\vartheta _{us}} = \sum\limits_{j \in \mN} {h_{us}^j/\left( {\tilde I_{us}^j + {\sigma ^2}} \right)}$ and ${\Gamma _s}\left( {{p_u}} \right) = \frac{{{\phi _u} + {\psi _u}{p_u}}}{{{{\log }_2}\left( {1 + {\vartheta _{us}}{p_u}} \right)}}$.
The objective function in~\eqref{eq:UPA_a} can now be approximated by $\sum\limits_{s \in \mS} {\sum\limits_{u \in {\mU_s}} {\Gamma _s}\left( {{p_u}} \right)  } $. With this position, it can be seen that the objective function and the constraint corresponding to each user's transmit power is now decoupled from each other. Therefore, the UPA problem in~\eqref{eq:UPA} can be approximated by $\sum\nolimits_{s \in \mS} {\left| {{\mU_s}} \right|}$ sub-problems, each optimizing the transmit power of a user $u\in \mU_s, s\in \mS$, and can be written as,
\begin{subequations} \label{eq:UPAapprox}
\begin{align} 
&\mathop {\min } {\sum\limits_{u \in {\mU_s}} {\Gamma _s}\left( {{p_u}} \right)  }  \\  \label{eq:UPAapprox_b}
&{\rm{s.t.}} \hspace{0.2cm} 0 < p_u \leq P_u.
\end{align}
\end{subequations}
Problem~\eqref{eq:UPAapprox} is still non-convex as the second-order derivative of the objective function with respect to (w.r.t) $p_u$, i.e., $\Gamma _s''\left( {{p_u}} \right)$, is not always positive. However, we can employ quasi-convex optimization technique to address problem~\eqref{eq:UPAapprox} based on the following lemma.

\begin{lemma}  \label{lem:quasi}
$\Gamma _s\left( {{p_u}} \right)$ is strictly quasi-convex in the domain defined in~\eqref{eq:UPAapprox_b}.
\end{lemma}
\begin{proof}
See Appendix.
\end{proof}
In general, a quasi-convex problem can be solved using the bisection method, which solves a convex feasibility problem in each iteration~\cite{boyd2004convex}. However, the popular interior cutting plane method for solving a convex feasibility problem requires $\m{O}\left( {{n^2}/{\varepsilon ^2}} \right)$ iterations, where $n$ is the dimension of the problem. We now propose to further reduce the complexity of the bisection method.  

Firstly, notice that a quasi-convex function achieves a local optimum at the diminishing point of the first-order derivative, and that any local optimum of a strictly quasi-convex function is the global optimum~\cite{bereanu1972quasi}. Therefore, based on Lemma~\ref{lem:quasi}, we can confirm that the optimal solution $p_u^*$ of problem~\eqref{eq:UPAapprox} either lies at the constraint border, i.e., $p_u^* = P_u$ or satisfies $\Gamma _s'\left( {{p_u^*}} \right) = 0$. It can be verified that $\Gamma _s'\left( {{p_u}} \right) = 0$ when,
\begin{equation} \label{eq:Omega}
\hspace{-0.05cm} {\Omega _s}\left( {{p_u}} \right) = {\psi _u}{\log _2}\left( {1 + {\vartheta _{us}}{p_u}} \right) - \frac{{{\vartheta _{us}}\left( {{\phi _u} + {\psi _u}{p_u}} \right)}}{{\left( {1 + {\vartheta _{us}}{p_u}} \right)\ln 2}} = 0. 
\end{equation}
Moreover, we have, ${\Omega _s'}\left( {{p_u}} \right) = \frac{{\vartheta _{us}^2\left( {{\phi _u} + {\psi _u}{p_u}} \right)}}{{{{\left( {1 + {\vartheta _{us}}{p_u}} \right)}^2}\ln 2}} > 0$, and ${\Omega _s}\left( 0 \right) =  - \frac{{{\vartheta _{us}}{\phi _u}}}{{\ln 2}} < 0$. This implies that ${\Omega _s}\left( {{p_u}} \right)$ is a monotonically increasing function and is negative at the starting point $p_u = 0$. Therefore, we can design a low-complexity bisection method that evaluates ${\Omega _s}\left( {{p_u}} \right)$ in each iteration instead of solving a convex feasibility problem, so as to obtain the optimal solution $p_u^*$, as presented in Algorithm~\ref{alg:bisection}.

\begin{algorithm} 
\caption{Bisection Method for Uplink Power Allocation} \label{alg:bisection}
\renewcommand{\Statex}{\item[\hphantom{\bfseries Step \arabic{ALG@line}.}]}
\begin{algorithmic}[1]
\State Calculate ${\Omega _s}\left( {{P_u}} \right)$ using~\eqref{eq:Omega}
\If{${\Omega _s}\left( {{P_u}} \right) \leq 0$}
	\State $p_u^* = P_u$
\Else 
	\State Set optimality tolerance $\epsilon > 0$
	\State Initialize $p_u' = 0$ and $p_u'' = P_u$
    \Repeat
    	\State Set $p_u^* = \left( {{p_u'} + {p_u''}} \right)/2$
        \If {${\Omega _s}\left( {{p_u^*}} \right) \leq 0$}
        	\State Set $p_u' = p_u^*$
        \Else
        	\State Set $p_u'' = p_u^*$
        \EndIf
    \Until $p_u'' - p_u' \leq \epsilon$
    \State Set $p_u^* = \left( {{p_u'} + {p_u''}} \right)/2$
\EndIf
\end{algorithmic}
\end{algorithm}
In Algorithm~\ref{alg:bisection}, if ${\Omega _s}\left( {{P_u}} \right) > 0$, the algorithm will terminate in exactly $\left\lceil {{{\log }_2}\left( {{P_u}/\epsilon } \right)} \right\rceil$ iterations. Let ${\mP^*} = \left\{ {p_u^*,u \in \mU} \right\}$ denote the optimal uplink transmit power policy for a given task offloading policy $\mX$. Denote now as $\Gamma \left( {\mX, {\mP^*}} \right)$ the objective value of problem~\eqref{eq:UPA} corresponding to $\mP^*$. 

\subsection{Computing Resource Allocation~(CRA)}
The CRA problem optimizes the second term on the RHS of~\eqref{eq:V_new} and is expressed as follows,
\begin{subequations} \label{eq:CRA}
\begin{align}  \label{eq:CRA_a}
&\mathop {\min }\limits_{\mF}\sum\limits_{s \in \mS} {\sum\limits_{u \in {\mU_s}} \frac{\eta_u}{f_{us}}  }   \\ \label{eq:CRA_b}
&{\rm{s.t.}} \hspace{0.2cm} \sum\nolimits_{u \in \mU} {{f_{us}}}  \le {f_s}, \forall s \in \mS.
\end{align}
\end{subequations}
Notice that the constraint in~\eqref{eq:CRA_b} is convex. Denote the objective function in~\eqref{eq:CRA_a} as $\Lambda \left( {\mX,\mF} \right)$; by calculating the second-order derivatives of $\Lambda \left( {\mX,\mF} \right)$ w.r.t. $f_{us}$, we have,
\begin{subequations}
\begin{align}
\frac{{{\partial ^2}\Lambda \left( {\mX,\mF} \right)}}{{\partial f_{us}^2}} &= \frac{{2{\eta _u}}}{{f_{us}^3}} > 0, \forall s \in \mS, u\in \mU_s, \\
\frac{{{\partial ^2}\Lambda \left( {\mX,\mF} \right)}}{{\partial {f_{us}}\partial {f_{vw}}}} &= 0,\forall \left( {u,s} \right) \ne \left( {v,w} \right).
\end{align}
\end{subequations}
It can be seen that the Hessian matrix of the objective function in~\eqref{eq:CRA_a} is diagonal with the strictly positive elements, thus it is positive-definite. Hence, \eqref{eq:CRA} is a convex optimization problem and can be solved using Karush-Kuhn-Tucker~(KKT) conditions. In particular, the optimal computing resource allocation $f_{us}^*$ is obtained as,
\begin{equation} \label{eq:f_opt}
f_{us}^* = \frac{{{f_s}\sqrt {{\eta _u}} }}{{\sum\nolimits_{u \in {\mU_s}} {\sqrt {{\eta _u}} } }}, \forall s \in \mS,
\end{equation}
and the optimal objective function is calculated as,
\begin{equation} \label{eq:Delta}
\Lambda \left( {\mX, {\mF^*}} \right) = \sum\limits_{s \in \mS} {\frac{1}{{{f_s}}}{{\left( {\sum\nolimits_{u \in {\mU_s}} {\sqrt {{\eta _u}} } } \right)}^2}}.
\end{equation}
\subsection{Joint Task Offloading Scheduling and Resource Allocation} 
In the previous sections, for a given task offloading decision $\mX$, we obtained the solutions for the radio and computing resources allocation. In particular, according to~\eqref{eq:prob_RA}, \eqref{eq:sys_util_new}, \eqref{eq:V_new}, and \eqref{eq:Delta}, we have,
\begin{equation} \label{eq:JX}
{J^*}\left( \mX \right) = \sum\limits_{s \in \mS} {\sum\limits_{u \in {\mU_s}} {{\lambda _u}\left( {\beta _u^t + \beta _u^e} \right)} }  - \Gamma \left( {\mX,{\mP^*}} \right) - \Lambda \left( {\mX,{\mF^*}} \right),
\end{equation}
where $\mP^*$ can be obtained through Algorithm~\ref{alg:bisection} and $\Lambda \left( {\mX,{\mF^*}} \right)$ can be calculated using the closed-form expression in~\eqref{eq:Delta}. Now, using~\eqref{eq:JX}, we can rewrite the TO problem in~\eqref{eq:master} as,
\begin{subequations} \label{eq:setmax}
\begin{align} \label{eq:setmax_a}
&\hspace{-1cm} \mathop {\max }\limits_{\mX}\sum\limits_{s \in \mS} {\sum\limits_{u \in {\mU_s}} {{\lambda _u}\left( {\beta _u^t + \beta _u^e} \right)} }  - \Gamma \left( {\mX,{\mP^*}} \right) - \Lambda \left( {\mX,{\mF^*}} \right)  \\ \label{eq:setmax_b}
\rm{s.t.} \hspace{0.2cm} &{x_{us}^j} \in \left\{ {0,1} \right\}, \forall u \in \mU, s \in \mS, j \in \mN,\\ \label{eq:setmax_c}
&\sum\limits_{s \in \mS} {\sum\limits_{j \in \mN} {x_{us}^j} }  \le 1, \forall u \in \mU, \\ \label{eq:setmax_d}
&\sum\limits_{u \in \mU} {{x_{us}^j}}  \le 1, \forall s \in \mS, j \in \mN.
\end{align}
\end{subequations}
Problem~\eqref{eq:setmax} consists in maximizing a set function ${J^*}\left( \mX \right)$ w.r.t $\mX$ over the ground set $\mG$ defined by~\eqref{eq:setmax_b}, and the constraints in~\eqref{eq:setmax_c} and \eqref{eq:setmax_d} define two matroids over $\mG$. Due to the NP-hardness of such problem~\cite{lee2009non}, designing efficient algorithms that guarantee the optimal solution still remains an open issue. In general, a brute-force method using exhaustive search would require evaluating $2^n$ possible task offloading scheduling decisions, where $n = S\times U \times N$, which is clearly not a practical approach. 

\begin{routine}
\caption{\emph{remove} and \emph{exchange} operations}\label{routine:remove_exchange}
\begin{algorithmic}[1]
\Statex ${{remove}}\left( {\mX,x_{us}^j} \right)$
\State Set $\mX \leftarrow \mX \backslash \left\{ x_{us}^j\right\}$
\State Output: $\mX$
\vspace{.2cm}
\Statex ${{exchange}}\left( {\mX,x_{us}^j} \right)$
\For{$w \in \mS, i \in \mN$}
	\State $\mX \leftarrow \mX \backslash \left\{ x_{uw}^i\right\}$
\EndFor

\For{$v \in \mU$}
	\State $\mX \leftarrow \mX \backslash \left\{ x_{vs}^j\right\}$
\EndFor
\State Set $\mX \leftarrow \mX \cup \left\{ x_{us}^j\right\}$

\State Output: $\mX$
\end{algorithmic}
\end{routine}

To overcome the aforementioned drawback, we propose a low-complexity heuristic algorithm that can find a local optimum to problem~\eqref{eq:setmax} in polynomial time. Specifically, our algorithm starts with an empty set $\mX = \emptyset$ and repeatedly performs one of the local operations, namely the \emph{remove} operation or the \emph{exchange} operation, as described in Routine~\ref{routine:remove_exchange}, if it improves the set value $J^*(\mX)$. As we are dealing with two matroid constraints, the \emph{exchange} operation involves adding one element from outside of the current set and dropping up to $2$ elements from the set, so as to comply with the constraints. In summary, our proposed heuristic algorithm for task offloading scheduling is presented in Algorithm~\ref{alg:task_offloading}.

\begin{algorithm} 
\caption{Heuristic Task Offloading Scheduling} \label{alg:task_offloading}
\renewcommand{\Statex}{\item[\hphantom{\bfseries Step \arabic{ALG@line}.}]}
\begin{algorithmic}[1]
\State Initialize: $\mX = \emptyset$ 

\State Find $x_{kw}^i = \mathop {\arg \max }\limits_{x_{us}^j,j \in \mN,s \in \mS,u \in \mU} {J^*}\left( {\left\{ {x_{us}^j} \right\}} \right)$


\State Set $\mX \leftarrow \left\{ x_{kw}^i \right\}$

\If{\hspace{-0.1cm} there exists $x_{us}^j \in \mX$ such that ${J^*}\left( {{{remove}}\left( {\mX,x_{us}^j} \right)} \right) > \left( {1 + \frac{\epsilon }{{{n^2}}}} \right){J^*}\left( \mX \right)$}
	
    \State Set $\mX \leftarrow {{{remove}}\left( {\mX,x_{us}^j} \right)}$
	\State Go back to step 4
    
\ElsIf{there exists $x_{us}^j \in \mG \backslash \mX$ such that ${J^*}\left( {{{exchange}}\left( {\mX,x_{us}^j} \right)} \right) > \left( {1 + \frac{\epsilon }{{{n^2}}}} \right){J^*}\left( \mX \right)$}
	    \State Set $\mX \leftarrow {{{exchange}}\left( {\mX,x_{us}^j} \right)}$
		\State Go back to step 4
\EndIf
\State Output: $\mX$
\end{algorithmic}
\end{algorithm}

\remark{(Complexity Analysis of Algorithm~\ref{alg:task_offloading}) 
Parameter $\epsilon > 0$ in Algorithm~\ref{alg:task_offloading} is any value such that $\frac{1}{\epsilon}$ is at most a polynomial in $n$. Let ${\rm{Opt}}\left( \mG \right)$ be the optimal value of problem~\eqref{eq:setmax} over the ground set $\mG$. It is easy to see that ${{\rm{J}}^*}\left( {\left\{ {x_{kw}^i} \right\}} \right) \le {\rm{Opt}}\left( \mG \right)/n$ where ${x_{kw}^i}$ is the element with the maximum ${{\rm{J}}^*}\left( {\left\{ {x_{us}^j} \right\}} \right)$ over all elements of $\mG$. Let $t$ be the number of iterations for 
Algorithm~\ref{alg:task_offloading}. Since after each iteration the value of the function increases by a factor of at least $\left( {1 + \frac{\epsilon }{{{n^2}}}} \right)$, we have ${\left( {1 + \frac{\epsilon }{{{n^2}}}} \right)^t} \le n$, and thus $t = \m{O}\left( {\frac{1}{\epsilon}{n^2}\log n} \right)$. Note that the number of queries needed to calculate the value of the objective function in each iteration is at most $n$. Therefore, the running time of Algorithm~\ref{alg:task_offloading} is $\m{O}\left( {\frac{1}{\epsilon}{n^3}\log n} \right)$, which is polynomial in $n$.
}

\remark{(Solution of JTORA)
Let $\mX^*$ be the output of Algorithm~\ref{alg:task_offloading}. The corresponding solutions $\mP^*$ for the uplink power allocation and $\mF^*$ for computing resource sharing can be obtained using Algorithm~\ref{alg:bisection} and the closed-form expression in~\eqref{eq:f_opt}, respectively, by setting $\mX = \mX^*$. Thus, the local optimal solution for the JTORA problem is $\left( {{\mX^*},{\mP^*},{\mF^*}} \right)$. While characterizing the degree of suboptimality of the proposed solution is a non-trivial task---mostly due to the combinatorial nature of the task offloading decision and the nonconvexity of the original UPA problem---in the next section we will show via numerical results that our heuristic algorithm performs closely to the optimal solution using exhaustive search method.
}


\section{Performance Evaluation} \label{sec:results}
In this section, simulation results are presented to evaluate the performance of our proposed heuristic joint task offloading scheduling and resource allocation strategy, referred to as hJTORA. We consider a multi-cell cellular system consisting of multiple hexagonal cells with a BS in the center of each cell. The neighboring BSs are set $1~\rm{km}$ apart from each other. We assume that both the users and BSs use a single antenna for uplink transmissions. The channel gains are generated using a distance-dependent path-loss model given as $L\left[ {{\rm{dB}}} \right] = 140.7 + 36.7{\log _{10}}{d_{\left[ {{\rm{km}}} \right]}}$, and the log-normal shadowing variance is set to $8~\rm{dB}$. In most simulations, if not stated otherwise, we consider $S = 7$ cells and the users' maximum transmit power set to $P_u = 20~\rm{dBm}$. In addition, the system bandwidth is set to $B = 20~\rm{MHz}$ and the background noise power is assumed to be $\sigma^2 = -100~\rm{dBm}$. 

\begin{table*}[h!]
\renewcommand{\arraystretch}{1.5}
\caption{Runtime Comparison Among Competing Schemes}\label{tab:runtime}
\centering
\begin{tabular}{|c|c|c|c|c|c|}
\hline
 &{\textbf{IOJRA}} & {\textbf{GOJRA}} & {\textbf{DORA}} & {\textbf{hJTORA}} & {\textbf{Exhaustive}} \\
\hline
\textbf{Runtime~[\rm{ms}]} & $0.2 \pm 0.03$ & $1.8 \pm 0.2$ & $6.8 \pm 0.22$ & $19.3 \pm 0.7$ & $1,923 \pm 1.4$ \\
\hline
\end{tabular}
\end{table*}

In terms of computing resources, we assume the CPU capability of each MEC server and of each user to be $f_s = 20~\rm{GHz}$ and $f_u^l = 1~\rm{GHz}$, respectively.
According to the realistic measurements in~\cite{miettinen2010energy}, we set the energy coefficient $\kappa$ as $5 \times 10^{-27}$. 
For computation task, we consider the face detection and recognition application for airport security and surveillance~\cite{soyata2012cloud} which can be highly benefit from the collaboration between mobile devices and MEC platform. Unless otherwise stated, we choose the default setting values as $d_u = 420~\rm{KB}$, ${c_u} = 1000~\rm{Megacycles}$ (following~\cite{soyata2012cloud,chen2016efficient}), $\beta_u^t = 0.2$, $\beta_u^e = 0.8$, and $\lambda_u = 1$, $\forall u \in \mU$. In addition, the users are placed in random locations, with uniform distribution, within the coverage area of the network, and the number of sub-bands $N$ is set equal to the number of users per cell. 
We compare the system utility performance of our proposed hJTORA strategy against the following approaches.
\begin{itemize}[leftmargin=*]
\item \emph{Exhaustive}: This is a brute-force method that finds the optimal offloading scheduling solution via exhaustive search over $2^n$ possible decisions; since the computational complexity of this method is very high, we only evaluate its performance in a small network setting.
\item \emph{Greedy Offloading and Joint Resource Allocation~(GOJRA)}: All tasks (up to the maximum number that can be admitted by the BSs) are offloaded, as in~\cite{sardellitti2015joint}. In each cell, offloading users are \emph{greedily} assigned to sub-bands that have the highest channel gains until all users are admitted or all the sub-bands are occupied; we then apply joint joint resource allocation across the BSs as proposed in Sect.~\ref{sec:res_alloc}-A,~B.
\item \emph{Independent Offloading and Joint Resource Allocation~(IOJRA)}: Each user is randomly assigned a sub-band from its home BS, then the users independently make offloading decision~\cite{zhang2015collaborative}; joint resource allocation is employed.
\item \emph{Distributed Offloading and Resource Allocation~(DORA)}: Each BS independently makes joint task offloading decisions and resource allocation for users within its cell~\cite{lyu2016multi}.
\end{itemize}

\subsection{Suboptimality of Algorithm~\ref{alg:task_offloading}}
\begin{figure}
\centering
\hspace{-.5cm}
\includegraphics[width=.8\textwidth]{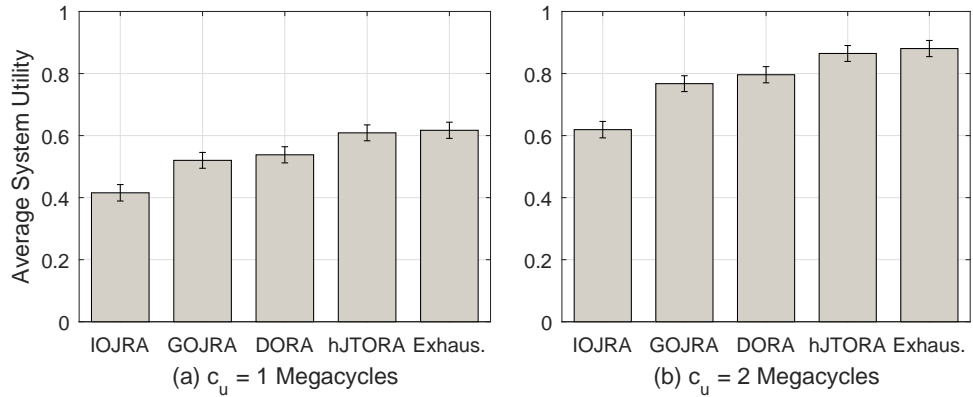}
\caption{Comparison of average system utility with $95\%$ confidence intervals.}\label{fig:subopt}
\end{figure}




Firstly, to characterize the suboptimality of our proposed hJTORA solution, we compare its performance with the optimal solution obtained by the \emph{Exhaustive} method, and then with the three other described baselines. Since the \emph{Exhaustive} method searches over all possible offloading scheduling decisions, 
its runtime is extremely long for a large number of variables; hence, we carry out the comparison in a small network setting with $U = 6$ users uniformly placed in the area covered by $S = 4$ cells, each having $N=2$ sub-bands. 
We randomly generate $500$ large-scale fading (shadowing) realizations and the average system utilities (with $95\%$ confident interval) of different schemes are reported in Fig.~\ref{fig:subopt}(a,b) when we set $c_u = 1000$ and $2000~\rm{Megacycles}$, respectively. It can be seen that the proposed hJTORA performs very closely to that of the optimal \emph{Exhaustive} method, while significantly outperforms the other baselines. In both cases, the hJTORA algorithm achieves an average system utility within $2\%$ that of the \emph{Exhaustive} algorithm, while providing upto $13\%$, $17\%$, and $47\%$ gains over the DORA, GOJRA, and IOJRA schemes, respectively. Additionally, in Table~\ref{tab:runtime}, we report the average runtime per simulation drop of different algorithms, running on a Windows 7 desktop with $3.6~\rm{GHz}$ CPU and $16~\rm{GB}$ RAM. It can be seen that the \emph{Exhaustive} method takes very long time, about $100\times$ longer than the hJTORA algorithm for such a small network. The DORA algorithm runs slightly faster than hJTORA while IOJRA and GOJRA requires the lowest runtimes.



\subsection{Effect of Number of Users}

\begin{figure}
\centering
\hspace*{-.2cm}
\includegraphics[width=0.90\textwidth]{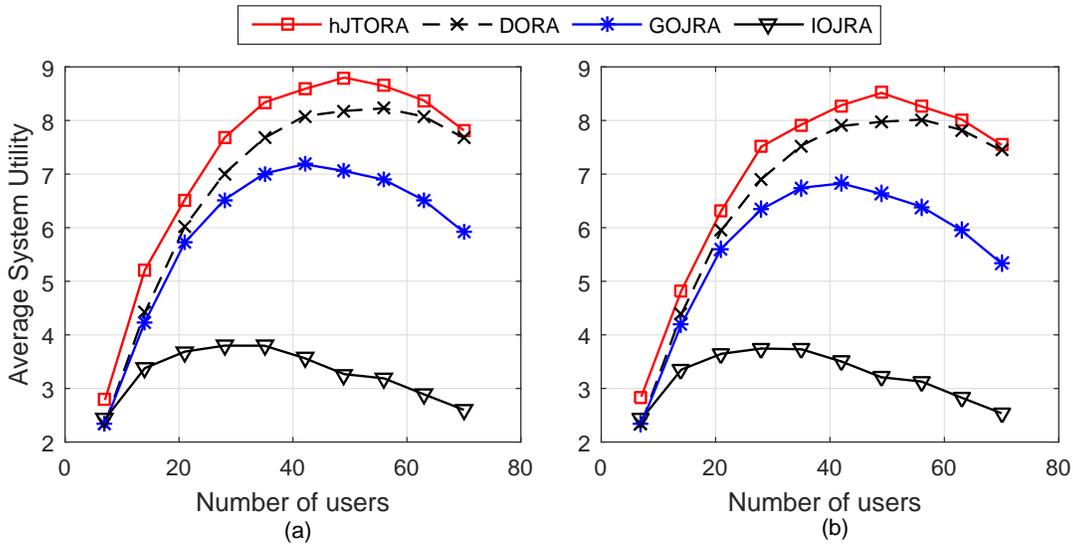}
\caption{Comparison of average system utility against different number of users, evaluated with two task workload distributions: (a) \textit{uniform}, $c_u = 1000~\rm{Megacycles}, \forall u \in \mU$, (b) \textit{non-uniform}, $c_u = 500~\rm{Megacycles}, \forall u$ in cells $\left\{ 1,3,5,7 \right\}$ and $c_u = 2000~\rm{Megacycles}, \forall u$ in cells $\left\{ 2,4,6 \right\}$.}\label{fig:nUsers}
\end{figure}

We now evaluate the system utility performance against different number of users wishing to offload their tasks, as shown in Fig.~\ref{fig:nUsers}(a,b). In particular, we vary the number of users per cell from $1$ to $10$ and perform the comparison in two scenarios with different task workload distribution: (a) \emph{uniform}, $c_u = 1000~\rm{Megacycles}$, and (b) \emph{non-uniform}, $c_u = 500~\rm{Megacycles}, \forall u$ in cells $\left\{ 1,3,5,7 \right\}$ and $c_u = 2000~\rm{Megacycles}$. Note that the number of sub-bands $N$ is set equal to the number of users per cell, thus the bandwidth allocated for each user decreases when there are more users in the system. Observe from Fig.~\ref{fig:nUsers}(a,b) that hJTORA always performs the best, and that the performance of all schemes significantly increases when the tasks' workload increases. This is because when the tasks require more computation resources the users will benefit more from offloading them to the MEC servers. We also observe in both scenarios that, when the number of users is small, the system utility increases with the number of users; however, when the number of users exceeds some threshold, the system utility starts to decrease. This is because when there are many users competing for radio and computing resources for offloading their tasks, the overheads of sending the tasks and executing them at the MEC servers will be higher, thus degrading the offloading utility.

\subsection{Effect of Task Profile}
Here, we evaluate the system utility performance w.r.t. to the computation tasks' profiles in terms of input size $d_u$'s and workload $c_u$'s. In Fig.~\ref{fig:taskprofile}(a,b), we plot the average system utility of the four competing schemes at different values of $c_u$ and $d_u$, respectively. It can be seen that the average system utilities of all schemes increase with task workload and decrease with the task input size. This implies that the tasks with small input sizes and high workloads benefit more from offloading than those with large input sizes and low workloads do. Moreover, we observe that the performance gains of the proposed hJTORA scheme over the baselines also follow the similar trend, i.e., increasing with task workloads and decreasing with task input size.  

\subsection{Effect of Users' Preferences}
Figure~\ref{fig:pTime} shows the average time and energy consumption of all the users when we increase the users' preference to time, $\beta_u^t$'s,  between $0.1$ and $0.9$ while at the same time decrease the users' preference to energy as $\beta_u^e = 1 - \beta_u^t, \forall u \in \mU$. It can be seen that the average time consumption decreases when $\beta_u^t$ increases, at the cost of higher energy consumption. In addition, when $U = 21$, the users experience a larger average time and energy consumption than in the case when $U = 14$. This is because when there are more users competing for the limited resources, the probability that a user can benefit from offloading its task is lower.

\begin{figure}
\centering
\hspace*{-.3cm}
\includegraphics[width=0.70\textwidth]{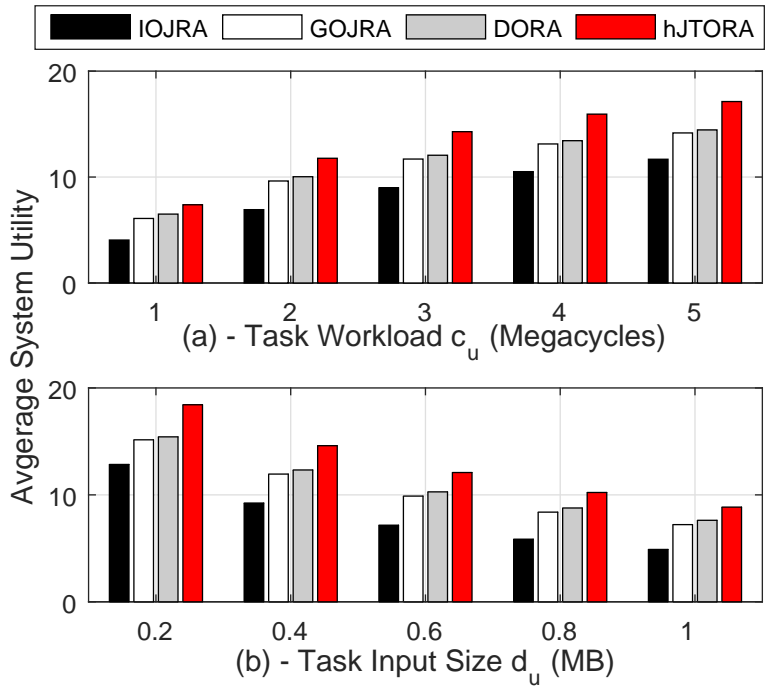}
\caption{Comparison of average system utility against (a) different task workloads, and (b) different sizes of task input; with $U = 28$.}\label{fig:taskprofile}
\end{figure}

\begin{figure}
\centering
\includegraphics[width=0.7\textwidth]{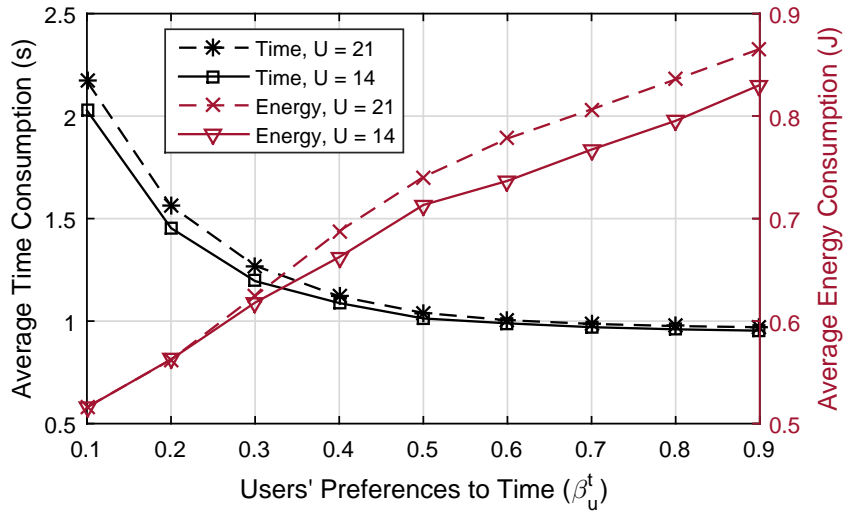}
\caption{Average time and energy consumption of all users obtained using hJTORA, with the number of users being $U = 14$ and $21$.}\label{fig:pTime}
\end{figure}

\subsection{Effect of Inter-cell Interference Approximation}
To test the effect of the approximation to model the inter-cell interference as in~\eqref{eq:Iapprox} in Sect.~\ref{sec:res_alloc}-A, we compare the results of the hJTORA solution to calculate the system utility using the approximated expression versus using the exact expression of the inter-cell interference. Figure~\ref{fig:approx} shows the system utility when the users' maximum transmit power $P_u$'s vary between $0$ and $35~\rm{dBm}$. It can be seen that the performance obtained using the approximation is almost identical to that of the exact expression when $P_u$ is below $25~\rm{dBm}$, while an increasing gap appears when $P_u > 25~\rm{dBm}$. However, as specified in LTE standard, 3GPP TS36.101 section 6.2.3\footnote{Refer to: 3GPP TS36.101, V14.3.0, Mar.~2017}, the maximum UE transmit power is $23~\rm{dBm}$; hence, we can argue that the proposed approximation can work well in practical systems. 

\begin{figure}
\centering
\includegraphics[width=0.6\textwidth]{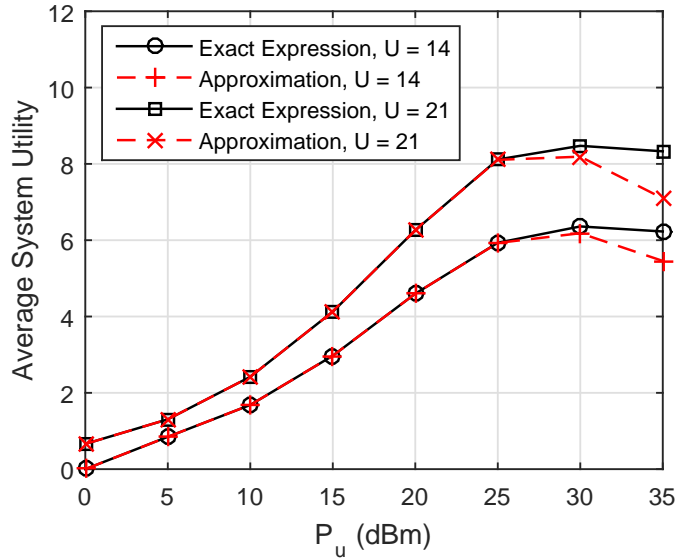}
\caption{Average system utility obtained by hJTORA solution with exact expression and approximation of the inter-cell interference.}\label{fig:approx}
\end{figure}

\section{Conclusions} \label{sec:conclusion}
We proposed a holistic strategy for a joint task offloading and resource allocation in a multi-cell Mobile-Edge Computing~(MEC) network. The underlying optimization problem was formulated as a Mixed-Integer Non-linear Program~(MINLP), which is NP-hard. Our approach decomposes the original problem into a Resource Allocation~(RA) problem with fixed task offloading decision and a Task Offloading~(TO) problem that optimizes the optimal-value function corresponding to the RA problem. We further decouple the RA problem into two independent subproblems, namely the uplink power allocation and the computing resource allocation, and address them using quasi-convex and convex optimization techniques, respectively. Finally, we proposed a novel heuristic algorithm that achieves a suboptimal solution for the TO problem in polynomial time. Simulation results showed that our heuristic algorithm performs closely to the optimal solution and significantly improves the average system offloading utility over traditional approaches. 

\section*{Appendix}
Firstly, it is straightforward to verify that ${\Gamma _s}\left( {{p_u}} \right)$ is twice differentiable on $\mathbb{R}$. We now check the second-order condition of a strictly quasi-convex function, which requires that a point $p$ satisfying ${\Gamma_s'}\left( {{p}} \right) = 0$ also satisfies ${\Gamma_s'}\left( {{p}} \right) > 0$~\cite{boyd2004convex}. 

The first-order and second-order derivatives of ${\Gamma _s}\left( {{p_u}} \right)$ can be calculated, respectively, as,

\begin{equation} \label{eq:1st_derivative}
{\Gamma _s'}\left( {{p_u}} \right) = \frac{{{\psi _u}{C_u}\left( {{p_u}} \right) - \frac{{{\vartheta _{us}}{D_u}\left( {{p_u}} \right)}}{{{A_u}\left( {{p_u}} \right)\ln 2}}}}{{C_u^2\left( {{p_u}} \right)}},
\end{equation}
and
\begin{equation} \label{eq:2nd_derivative}
{\Gamma _s''}\left( {{p_u}} \right) = \frac{{{\vartheta _{us}}\left[ {{G_{us}}\left( {{p_u}} \right){C_{us}}\left( {{p_u}} \right) + 2{\vartheta _{us}}{D_{us}}\left( {{p_u}} \right)/\ln 2} \right]}}{{A_{us}^2\left( {{p_u}} \right)C_{us}^3\left( {{p_u}} \right)\ln 2}},
\end{equation}
in which, 
\begin{subequations}
\begin{align}
{A_{us}}\left( {{p_u}} \right) &= 1 + {\vartheta _{us}}{p_u}, \\
{C_{us}}\left( {{p_u}} \right) &= {\log _2}\left( {1 + {\vartheta _{us}}{p_u}} \right), \\
{D_{us}}\left( {{p_u}} \right) &= {\phi _u} + {\psi _u}{p_u}, \\
{G_{us}}\left( {{p_u}} \right) &= {\vartheta _{us}}{D_{us}}\left( {{p_u}} \right) - 2{\psi _u}{A_{us}}\left( {{p_u}} \right).
\end{align}
\end{subequations}
Suppose that ${{\bar p}_u} \in \left( {0,{P_u}} \right]$; to satisfy $\Gamma _s'\left( {{\bar p_u}} \right) = 0$, it must hold that,
\begin{equation}
{\Omega _s}\left( {{{\bar p}_u}} \right) = {\psi _u}{\log _2}\left( {1 + {\vartheta _{us}}{{\bar p}_u}} \right) - \frac{{{\vartheta _{us}}\left( {{\phi _u} + {\psi _u}{{\bar p}_u}} \right)}}{{\left( {1 + {\vartheta _{us}}{{\bar p}_u}} \right)\ln 2}} = 0.
\end{equation}
By substituting $\bar p_u$ into~\eqref{eq:2nd_derivative}, we obtain, 
\begin{equation}
{\Gamma_s''}\left( {{{\bar p}_u}} \right) = \frac{{\vartheta _{us}^3D_{us}^2\left( {{{\bar p}_u}} \right)}}{{A_{us}^2\left( {{{\bar p}_u}} \right)C_{us}^3\left( {{{\bar p}_u}} \right){\psi _u}{{\ln }^2}2}}.
\end{equation}
It can be easily verified that both $\vartheta_{us}$ and ${D_{us}^2\left( {{{\bar p}_u}} \right)}$ are strictly positive $\forall {{\bar p}_u} \in \left( {0,{P_u}} \right]$. Hence, ${\Gamma_s''}\left( {{{\bar p}_u}} \right) > 0$, which confirms that ${\Gamma_s}\left( {{{p}_u}} \right)$ is a strictly quasi-convex function in $\left( {0,{P_u}} \right]$. 

\balance

\bibliographystyle{ieeetr}\small
\bibliography{2017_JSAC}

\end{document}